\documentclass[a4paper,USenglish,autoref]{lipics-v2019}
\usepackage{amsmath,amsthm,amssymb}
\usepackage[T1]{fontenc}
\usepackage{xspace}
\usepackage[dvipsnames]{xcolor}
\usepackage{mathtools}

\usepackage{cite}

\usepackage{algorithm}
\usepackage[noend]{algpseudocode}
\hideLIPIcs
\nolinenumbers

\theoremstyle{plain}

\definecolor{ablue}{rgb}{0.3,0.4,0.8}
\definecolor{ared}{rgb}{0.95,0.4,0.4}
\definecolor{agreen}{rgb}{0,0.5,0.25}
\definecolor{ayellow}{rgb}{0.95,0.85,0.3}

\newcommand{\Reals}{\mathbb{R}}

\newcommand{\Hyp}{\mathbb{H}}

\newcommand{\cG}{\mathcal{G}}

\newcommand{\cC}{\mathcal{C}}
\newcommand{\cP}{\mathcal{P}}

\newcommand{\eps}{\varepsilon}

\newcommand\eqdef{\mathrel{\overset{\makebox[0pt]{\mbox{\normalfont\tiny\sffamily def}}}{=}}}

\DeclareMathOperator{\dist}{dist}

\newcommand{\len}{\mathbf{length}}
\newcommand{\an}{\sphericalangle}

\DeclareMathOperator{\poly}{poly}

\DeclarePairedDelimiter{\norm}{\lVert}{\rVert}

\newcommand{\Cr}{\mathsf{Cr}}
\newcommand{\End}{\mathsf{End}}

\renewcommand{\leq}{\leqslant}
\renewcommand{\geq}{\geqslant}
\renewcommand{\rho}{\varrho}

\newcommand{\etal}{\emph{et~al.}}

\newcommand{\IS}{\textsc{Independent Set}\xspace}
\newcommand{\DS}{\textsc{Dominating Set}\xspace}
\newcommand{\ST}{\textsc{Steiner Tree}\xspace}
\newcommand{\HC}{\textsc{Hamiltonian Cycle}\xspace}
\newcommand{\HTSP}{\textsc{Hyperbolic TSP}\xspace}
\newcommand{\BTSP}{\textsc{Hyperbolic Path Cover}\xspace}
\newcommand{\ETSP}{\textsc{Euclidean TSP}\xspace}
\newcommand{\TSPlong}{Traveling Salesman Problem\xspace}

\newcommand{\NP}{\mbox{\ensuremath{\mathsf{NP}}}\xspace}

\newcommand{\defproblem}[3]{
\begin{quotation}
\noindent
  \textsc{#1}  \\
  {\bf{Input:}} #2 \\
  {\bf{Question:}} #3
\end{quotation}
}

\makeatletter
\DeclareFontFamily{U}{tipa}{}
\DeclareFontShape{U}{tipa}{m}{n}{<->tipa10}{}
\newcommand{\arc@char}{{\usefont{U}{tipa}{m}{n}\symbol{62}}}%

\newcommand{\arc}[1]{\mathpalette\arc@arc{#1}}

\newcommand{\arc@arc}[2]{%
  \sbox0{$\m@th#1#2$}%
  \vbox{
    \hbox{\resizebox{\wd0}{\height}{\arc@char}}
    \nointerlineskip
    \box0
  }%
}
\makeatother

\makeatletter
\DeclareRobustCommand{\bfseries}{%
  \not@math@alphabet\bfseries\mathbf
  \fontseries\bfdefault\selectfont
  \boldmath
}
\makeatother
\title{A quasi-polynomial algorithm for well-spaced hyperbolic TSP}

\author{S\'andor Kisfaludi-Bak}{Max Planck Institute for Informatics, Germany}{sandor.kisfaludi-bak@mpi-inf.mpg.de}{}{}

\authorrunning{S\'andor Kisfaludi-Bak} %mandatory. First: Use abbreviated first/middle names. Second (only in severe cases): Use first author plus 'et. al.'

\Copyright{S\'andor Kisfaludi-Bak}%mandatory, please use full first names. LIPIcs license is "CC-BY";  http://creativecommons.org/licenses/by/3.0/

%TODO mandatory: Please choose ACM 2012 classifications from https://dl.acm.org/ccs/ccs_flat.cfm

% mandatory: Please provide 1-5 keywords
\EventEditors{John Q. Open and Joan R. Access}
\EventNoEds{2}
\EventLongTitle{42nd Conference on Very Important Topics (CVIT 2016)}
\EventShortTitle{CVIT 2016}
\EventAcronym{CVIT}
\EventYear{2016}
\EventDate{December 24--27, 2016}
\EventLocation{Little Whinging, United Kingdom}
\EventLogo{}
\SeriesVolume{42}
\ArticleNo{23}
\begin{document}

\maketitle

\begin{abstract}
We study the traveling salesman problem in the hyperbolic plane of Gaussian curvature $-1$. Let $\alpha$ denote the minimum distance between any two input points. Using a new separator theorem and a new rerouting argument, we give an $n^{O(\log^2 n)\max(1,1/\alpha)}$ algorithm for \HTSP. This is quasi-polynomial time if $\alpha$ is at least some absolute constant, and it grows to $n^{O(\sqrt{n})}$ as $\alpha$ decreases to $\log^2 n/\sqrt{n}$. (For even smaller values of $\alpha$, we can use a planarity-based algorithm of Hwang et al. (1993), which gives a running time of $n^{O(\sqrt{n})}$.)
\end{abstract}

\section{Introduction}
The \TSPlong (or TSP for short) is very widely studied in combinatorial optimization and computer science in general, with a long history. In the general formulation, we are given a complete graph $G$ with positive weights on its edges. The task is to find a cycle through all the vertices (i.e., a Hamiltonian cycle) of minimum weight. The first non-trivial algorithm (with running time $O(2^nn^2)$) was given by Held and Karp~\cite{HeldK61}, and independently by Bellman~\cite{Bellman-TSP}. The problem was among the first problems to be shown \NP-hard by Karp~\cite{Karp10}.

A very important case of TSP concerns metric weight functions, where the edge weights satisfy the triangle inequality. The problem has a $(3/2)$-approximation due to Christofides~\cite{Chr76}, which is still unbeaten. On the other hand, it is \NP-hard to approximate \textsc{Metric} TSP within a factor of $123/122$~\cite{KarpinskiLS15}. Fortunately, the problem is more tractable in low-dimensional geometric spaces. Arora~\cite{Arora98} and independently, Mitchell~\cite{Mitchell99} gave the first polynomial time approximation schemes (PTASes) for the low-dimensional \ETSP problem, where vertices correspond to points in $\Reals^d$ and the weights are defined by the Euclidean distance between the given points. The PTAS was later improved by Rao and Smith~\cite{RaoS98}, and after two decades, several more general approximation schemes are known. In particular, there is a PTAS in metric spaces of bounded doubling dimension by Bartal~\etal~\cite{TSPdoubling16}, and in metric spaces of negative curvature by Krauthgamer and Lee~\cite{KrauthgamerL06}. The PTAS of~\cite{KrauthgamerL06} applies in the hyperbolic plane.

Turning to the exact version of the problem in the geometric setting, we can again get significant improvements over the best known $O(2^n\poly(n))$ running time for the general version. In the Euclidean case, the first set of improved algorithms were proposed in the plane by Kann~\cite{Kann92} and by~Hwang~\etal~\cite{HwangCL93} with running time $n^{O(\sqrt{n})}$. Later, an algorithm in $\Reals^d$ with running time $n^{O(n^{1-1/d})}$ was given by Smith and Wormald~\cite{SmithW98}. The latest improvement to $2^{O(n^{1-1/d})}$ by De Berg~\etal~\cite{BergBKK18} came with a matching lower bound under the Exponential Time Hypothesis (ETH)~\cite{ethcite}. To our knowledge, the exact version of the problem in hyperbolic space has not been studied yet.

Given the history of the problem, the PTAS results and the Euclidean exact algorithm, one might expect that the hyperbolic case is very similar to the Euclidean, and a good hyperbolic TSP algorithm will have a running time of $n^{O(n^{\delta})}$ for some constant $\delta$. In this paper, we show that we can often get significantly faster algorithms. Let $\Hyp^2$ denote the hyperbolic plane of Gaussian curvature $-1$. The first hopeful sign is that $\Hyp^2$ exhibits special properties when it comes to intersection graphs. Recently, the present author has given quasi-polynomial algorithms for several classic graph problems in certain hyperbolic intersection graphs of ball-like objects~\cite{hyperbolic_inters}. The studied problems include \IS, \DS, \ST, \HC and several other problems that are \NP-complete in general graphs. Interestingly, a polynomial time algorithm was given for the \HC problem in hyperbolic unit disk graphs. The question arises whether a quasi-polynomial algorithm is available for TSP in $\Hyp^2$? Given that the best running times for \HC in unit disk graphs in $\Reals^2$ and for \ETSP are identical, perhaps even polynomial time is achievable for \HTSP?

Unfortunately, a quasi-polynomial algorithm is unlikely to exist for the general \HTSP problem: the lower bound of~\cite{frameworkpaper} in grids can be carried over to $\Hyp^2$, which rules out a $2^{o(\sqrt{n})}$ algorithm under the Exponential Time Hypothesis (ETH)~\cite{ethcite}. This however relies on embedding a grid-like structure in $\Hyp^2$ efficiently, which seems to be possible only if the points are densely placed. Since $\Hyp^2$ is locally Euclidean, it comes as no surprise that we cannot beat the Euclidean running time for dense point sets.

For this reason, we use a parameter measuring the density of the input point set. We say that the input point set $P$ is \emph{$\alpha$-spaced} if for any pair of distinct points $p,p'\in P$, we have that $\dist(p,p')\geq \alpha$. 
Our main contribution is the following theorem.

\begin{theorem}\label{thm:main}
Let $P$ be an $\alpha$-spaced set of points in the hyperbolic plane of curvature $-1$. Then the shortest traveling salesman tour of $P$ can be computed in $n^{O(\log^2 n) \cdot \max(1,1/\alpha)) }$ time.
\end{theorem}

Note that for $\alpha\geq 1$, this is a quasi-polynomial algorithm. In Section~\ref{sec:app_hyptsplower} we show that for very dense inputs, it is unlikely that our running time can be improved significantly: we prove that there is no $2^{o(\sqrt{n})}$ algorithm for point sets of spacing $\Theta(1/\sqrt{n})$, unless the Exponential Time Hypothesis (ETH) fails.

\paragraph*{Adapting algorithms from the Euclidean plane.}

Most algorithms for \ETSP are difficult to adapt to the hyperbolic setting. The majority of known subexponential algorithms for \ETSP (see~\cite{Kann92,SmithW98,BergBKK18}) are based on some version of the so-called \emph{Packing Property}~\cite{BergBKK18}. The property roughly states that for any disk $\delta$ of radius $r$ and any optimal tour $\tau$, the number of segments in $\tau$ of length at least $r$ that intersect $\delta$ is at most some absolute constant. This starting point is not available to us, since a direct adaptation of the Packing Property as stated above is false in~$\Hyp^2$. For example, we can create a regular $n$-gon where the length of each side is $c\log n$ for some constant $c$, and the inscribed circle has radius $r<c \log n$. The boundary of the $n$-gon is an optimal tour of its vertices, and the inscribed disk is intersected more than a constant times with tour segments of length at least $r$.

The only exact \ETSP algorithm that directly carries over to $\Hyp^2$ is the algorithm of Hwang, Chang and Lee~\cite{HwangCL93}, as it only relies on the fact that any optimal tour in the plane is crossing-free. Unfortunately, this algorithm has a running time of $n^{O(\sqrt{n})}$, which is far from our goal. Nonetheless, we can use this algorithm for the case when the point set $P$ has close point pairs, that is, when $\alpha\leq \log^2 n/\sqrt{n}$. This is discussed further in Section~\ref{sec:prelim}.

\paragraph*{Our techniques.}
To get a quasi-polynomial algorithm for $\alpha=\Omega(1)$, we need to prove our own separator theorem. The separator itself is fairly simple: it is a line segment of length $O(\log n)$. Due to the special properties of $\Hyp^2$, optimal tours cannot go ``around'' this segment. The difficulty is to show that the line segment is crossed only $O(\log n)$ times by an optimal tour. We show that having a pair of ``nearby''\footnote{The absolute distance of crossing edges cannot be bounded; we use a special definition of ``nearby''.} tour edges crossing a certain  neighborhood $R$ of the segment can be ruled out with a rerouting argument that is reminiscent of the proof of the Packing Property in~$\Reals^2$. This limits the number of segments crossing both $R$ and the segment to $O(\log n)$. All other tour edges crossing the segment must have an endpoint in $R$. Since $R$ is ``narrow'', it can contain at most $O(\log n)$ points from $P$, as $P$ is $\alpha$-spaced. These bounds together limit the number of tour edges crossing our segment to $O(\log n)$.
With the separator at hand, we use a standard divide-and-conquer algorithm to prove Theorem~\ref{thm:main}. For values $\alpha\leq \log^2 n/\sqrt{n}$, we suggest using the algorithm of Hwang~\etal~\cite{HwangCL93}.

\paragraph*{Computational model.}

As our input, we get a list of points $P$ with rational coordinates in the Poincar\'e disk model (which we briefly introduce in Section~\ref{sec:prelim}) and a rational number $x$. The goal is to decide if there is a tour of length at most $x$.

It is a common issue in computational geometry that one needs to be able to compare sums of distances. In geometric variants of TSP, this directly impacts the output, and unfortunately no method is known to tackle this in a satisfactory manner on a word-RAM machine. For this reason, most work in the area assumes that the computation is done on a real-RAM machine that can compute square roots exactly. Perhaps even less is known about comparing sums of distances in hyperbolic space. For this article, we work in a real-RAM that, in addition to taking square roots, is also capable of computing the natural logarithm~$\ln(.)$.

\section{Preliminaries}\label{sec:prelim}

\paragraph*{The hyperbolic plane and the Poincar\'e disk model.} 
Introducing the hyperbolic plane properly is well beyond the scope of this
section, but we list some important properties that we will be using. A detailed
exposition can be found in several textbooks~\cite{benedetti2012lectures,
thurston1997three,greenberg1993euclidean,RamsayRichtmyer}. %The following paragraphs borrow from~\cite{doktori}.

The hyperbolic plane $\Hyp^2$ is a homogeneous metric space with the
key property that the area and circumference of disks grows exponentially with the radius, that is, a disk of radius $r$ has area $4\pi \sinh^2(r/2)$ and circumference $2\pi\sinh(r)$. For $r>1$, both the area and circumference are $\Theta(e^{r})$.
On the other hand, a small neighborhood of any point in the hyperbolic plane
is very similar to a small neighborhood of a point in the
Euclidean plane. More precisely, the disk of radius $\eps$ around a point in
$\Hyp^2$ and $\Reals^2$ have a smooth bijective mapping that preserve
distances up to a multiplicative factor of $1+f(\eps)$, where $\lim_{\eps
\rightarrow 0} f(\eps)=0$.

The hyperbolic plane itself can be defined in many ways, but it is most convenient to take
some region of $\Reals^2$, and equip it with a custom metric. Such
definitions are also called \emph{models} of the hyperbolic plane. In this article, we use the Poincar\'e disk model for all of the figures.

The Poincar\'e disk model is the open unit disk of $\Reals^2$ equipped with the distance function
\[\dist(u,v)=\cosh^{-1}\left(1+2\frac{\norm{u-v}^2}{(1-\norm{u}^2)(1-\norm{v}^2)}\right),\]
where $\norm{.}$ is the Euclidean norm.\footnote{As $\cosh^{-1}(x)=\ln(x+\sqrt{x^2-1})$, the distance of two points in the Poincar\'e disk model with given Euclidean coordinates can be computed on a real-RAM machine which is capable of taking square roots and computing~$\ln(.)$.} The precise function here is irrelevant; we present the formula just as an example of defining a custom metric space.\footnote{If we need to calculate angles, curve length, and area, we should define the metric tensor instead: $ds^2=4\frac{\norm{dx}^2}{(1-\norm{x}^2)^2}$~\cite{CannonFKP97}.} We list some further properties of $\Hyp^2$ used in the article.

\begin{figure}[t]
\centering
\includegraphics[height=5.5cm]{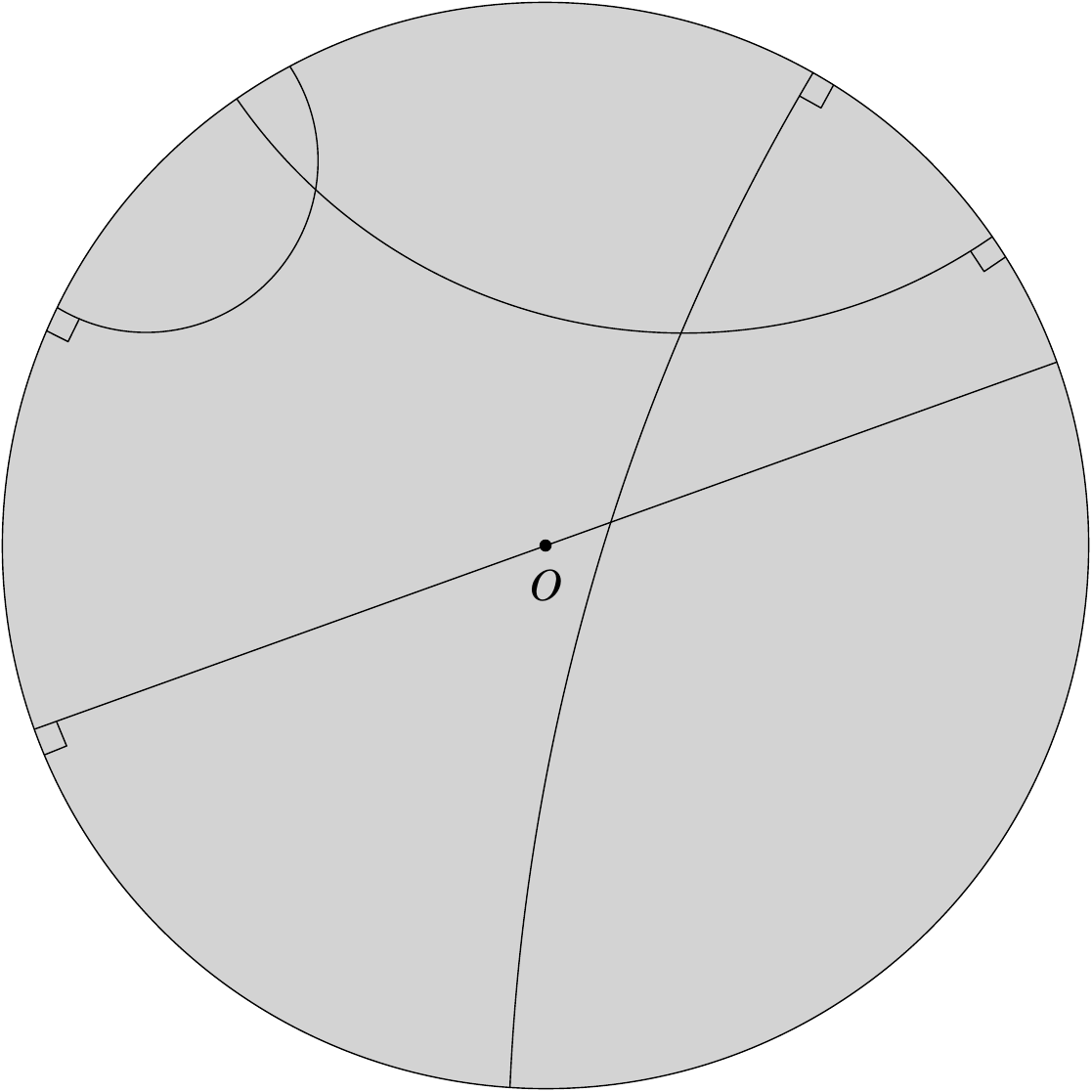}
\hspace{1.5cm}
\includegraphics[height=5.5cm,trim={3cm 3cm 0 0},clip]{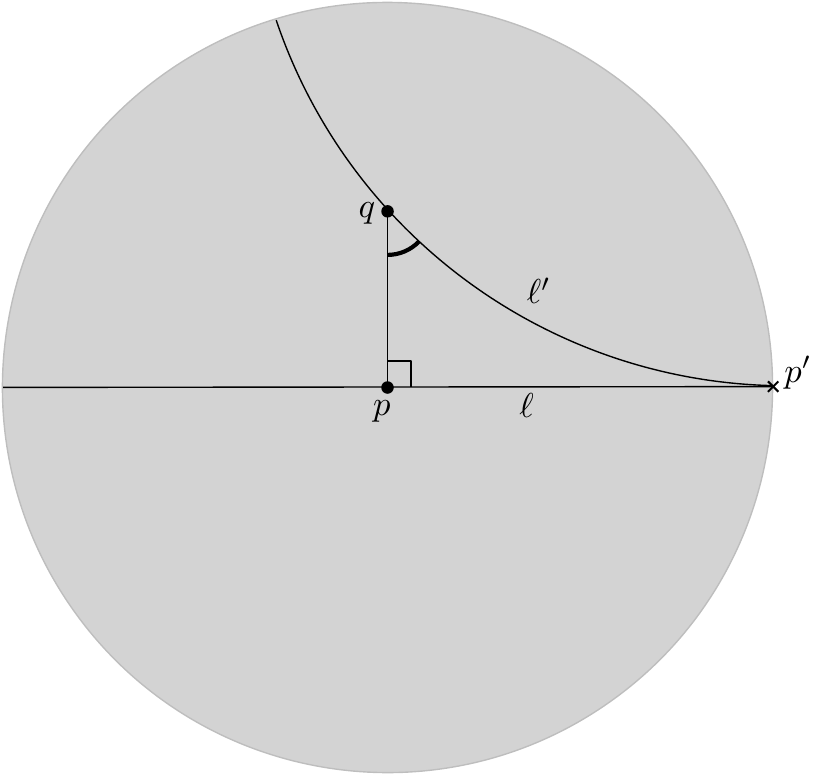}
\caption{Left: lines in the Poincar\'e model. Right: the angle of parallelism for the length $|pq|$.}\label{fig:hyperintro}
\end{figure}

\begin{itemize}
\item \emph{Lines, angles, and ideal points.}\\
In the Poincar\'e  disk model hyperbolic
lines appear as Euclidean circular arcs that are perpendicular to the unit
circle, as illustrated on the left of Figure~\ref{fig:hyperintro}. In particular, hyperbolic lines through the center of the disk are
diametrical segments of the unit disk. The model is \emph{conformal}, that is,
the angle of a pair of lines in $\Hyp^2$ is the same as the angle of
the corresponding arcs in $\Reals^2$. 
The points on the boundary of the disk are called \emph{ideal points}.
\item \emph{Angle of parallelism.}\\ Let $p,q\in \Hyp^2$ and let $\ell$ be the line through $p$ that is perpendicular to $pq$, and let $p'$ be an ideal point of $\ell$, see the right hand side of Figure~\ref{fig:hyperintro}. Let $\ell'$ be the line through $q$ and $p'$. Note that $\ell$ and $\ell'$ are disjoint lines in the open disk; they are called \emph{limiting parallels}. The angle $\an pqp'$ is called the \emph{angle of parallelism}, which only depends on the length of the segment $pq$ in the following way~\cite{RamsayRichtmyer}.
\begin{equation}
\tan(\an pqp')=\frac{1}{\sinh(|pq|)}
\end{equation}
\item \emph{Hypercycles or equidistant curves.}\\
The set of points at a given distance $\rho$ from a line $\ell$ is not a line, but it forms a \emph{hypercycle} in $\Hyp^2$. A hypercycle has two arcs, one on each side of $\ell$. In the Poincar\'e model, a hypercycle for a line $\ell$ consists of two circular arcs, ending at the same ideal points as $\ell$.
\item \emph{Optimal tours in $\Hyp^2$ and crossings.}\\ An optimal traveling salesman tour will consist of geodesics between pairs of input points, i.e., hyperbolic segments, just as in $\Reals^2$. Moreover, the triangle inequality implies that any self-crossing tour (where two segments $pp'$ and $qq'$ cross) can be shortened. Thus, optimal tours in $\Hyp^2$ are non-crossing.
\end{itemize}

\paragraph*{Getting a subexponential algorithm for all values of $\alpha$.}
We can give the following more general formulation of the result of~\cite{HwangCL93}.

\begin{theorem}[Hwang, Chang and Lee~\cite{HwangCL93}, stated generally]\label{thm:HwangGen}
Let $P$ be a set of $n$ points in $\Reals^2$, and let $w:\binom{P}{2}\rightarrow \Reals$ be a weight function on the (straight) segments defined by the point pairs. Suppose that the optimal TSP tour of $P$ with respect to $w$ is crossing-free. Then there is an algorithm to compute this optimal tour in $n^{O(\sqrt{n})}$ time.
\end{theorem}

We convert our initial point set $P$ in the Poincar\'e model to the Beltrami-Klein model of~$\Hyp^2$ to get a point set $P_{BK}$. In the Beltrami-Klein model, Euclidean segments inside the open unit disk are (geodesic) segments of~$\Hyp^2$. Since the optimal hyperbolic TSP tour is crossing-free, the tour in the Beltrami-Klein model is a polygon with vertex set $P_{BK}$. The hyperbolic distances can be used as weights on all segments with endpoints from $P_{BK}$, and we can apply Theorem~\ref{thm:HwangGen} to get an $n^{O(\sqrt{n})}$ algorithm regardless of the value of $\alpha$.

\section{A separator for \texorpdfstring{Hyperbolic TSP}{\HTSP}}

\paragraph*{Centerpoint and a separating line.} It has already been observed in~\cite{hyperbolic_inters} that for any set $P\subset \Hyp^2$
of $n$ points there exists a point $q\in \Hyp^2$ such that for any line $\ell$ through $q$ the two open half-planes with boundary $\ell$ both contain at most $\frac{2}{3}n$ points from $P$, that is, the line $\ell$ is a \emph{$2/3$-balanced separator} of $P$. Such a point $q$ is called a \emph{centerpoint} of $P$. It has been observed in~\cite{hyperbolic_inters} that given $P$, a centerpoint of $P$ can be computed using a Euclidean centerpoint algorithm, which takes linear time~\cite{DBLP:journals/dcg/JadhavM94}.

It is now easy to prove that we can find a balanced line separator that has a small neighborhood empty of input points. See Figure~\ref{fig:emptycone} for an illustraton.

\begin{lemma}\label{lem:sep}
Given a point set $P\subset \Hyp^2$, there exists a point $q\in \Hyp^2$ and there exists a line $\ell$ through $q$ such that $P$ is disjoint from the open double cone with center $q$, axis $\ell$ and half-angle $\frac{\pi}{2n}$. Any such line $\ell$ is a $2/3$-balanced separator of $P$, and given $P$, a suitable point $q$ and line $\ell$ can be found in linear time.
\end{lemma}
 
\begin{proof}
Let $q$ be a centerpoint of $P$. For each point $p\in P$, let $\ell_p$ be the line through $q$ and $p$. Since we have defined $n$ lines through $q$, there is a pair of consecutive lines $\ell_p,\ell_{p'}$ whose acute angle is at least $\pi/n$. Let $\ell$ be the angle bisector of $\ell_p$ and $\ell_{p'}$. Then $\ell$ clearly has the desired properties, and the centerpoint $q$, the lines $\ell_p,\ell_{p'}$ and the line $\ell$ can all be computed in linear time.
\end{proof}

We can extend Lemma~\ref{lem:sep} to get balance with respect to a subset $B\subseteq P$, that is, both half-planes bounded by $\ell$ would contain at most $\frac{2}{3}|B|$ points of $B$. One only needs to set $q$ to be the centerpoint of $B$ instead of $P$.

\begin{figure}[t]
\centering
\includegraphics[width=8cm]{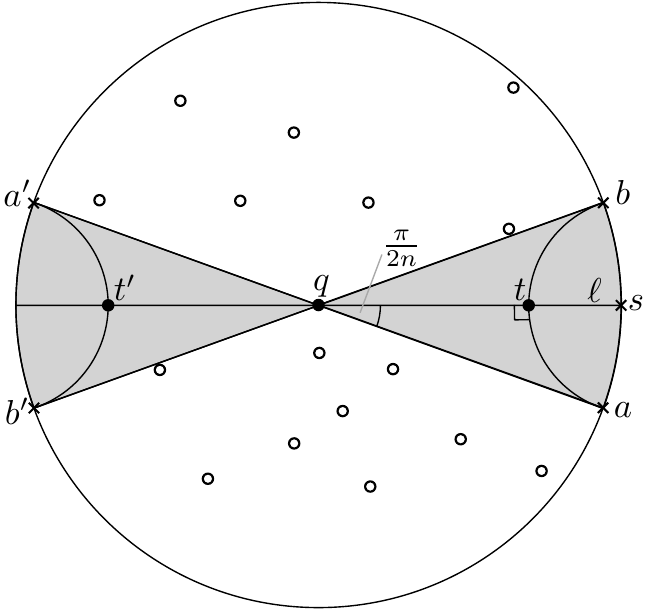}
\caption{The empty double cone with axis $\ell$.}\label{fig:emptycone}
\end{figure}

\paragraph*{Defining a region around the separator.}
From this point onwards, $q$ denotes a centerpoint of $P$, and $\ell$ is a line through $q$ with the properties from Lemma~\ref{lem:sep}. Let $C$ denote the double cone of center $q$, axis $\ell$ and half-angle $\frac{\pi}{2n}$, see Figure~\ref{fig:emptycone}. Note that by Lemma~\ref{lem:sep}, we have that $C\cap P=\emptyset$. Let $s$ be an ideal point of $\ell$, and let $a,b$ be ideal points on the boundary of $C$, such that $\an aqs = \an sqb = \frac{\pi}{2n}$. Let $t=ab \cap \ell$. Notice that $qta$ is a right-angle triangle with ideal point $a$, and it has angle $\frac{\pi}{2n}$ at $q$. Therefore, $\frac{\pi}{2n}$ is the angle of parallelism for the distance $|qt|$, and it satisfies
\begin{equation}\label{eq:qt_len}
\sinh(|qt|)=\frac{1}{\tan(\frac{\pi}{2n})}.
\end{equation}

The line $ab$ splits $\Hyp^2$ into two open half-planes: the side $H_q$ containing $q$, and the side $H_s$ that has $s$ on its boundary. Note that $H_s \subset C$, therefore $P\subset H_q$. Consequently, all segments of the tour are contained in $H_q$.
We mirror $a,b$ and $t$ to the point $q$; let $a',b'$ and $t'$ denote the resulting points respectively. By our earlier observation, the entire tour is contained in the geodesically convex region between the lines $ab$ and $a'b'$, and any tour segment intersecting $\ell$ will intersect it somewhere on the segment $tt'$.

\begin{figure}[t]
\centering
\includegraphics[width=10cm]{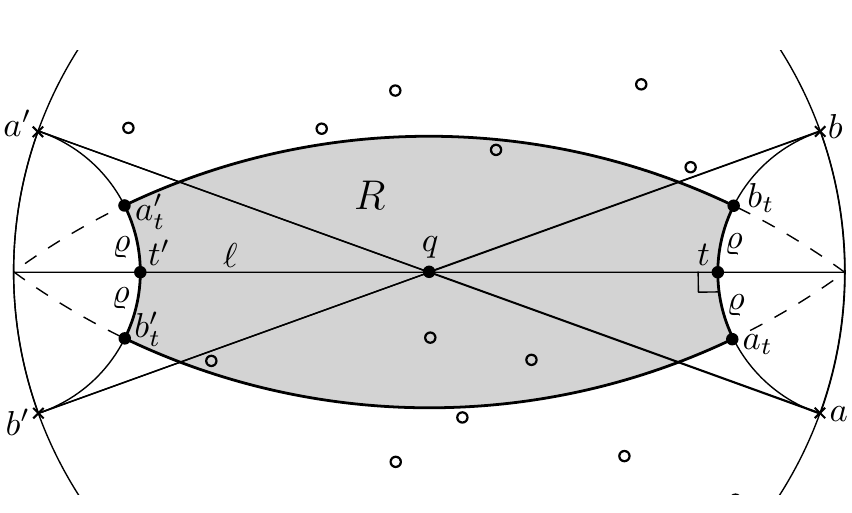}
\caption{The construction of the region $R$.}\label{fig:sep_region}
\end{figure}

Let $a_t$ and $b_t$ be the points on $ab$ at distance $\rho$ from $t$, where $\rho\in (0,\alpha/2)$ is a suitable number that will be defined later. Let $a'_t$ and $b'_t$ denote the analogous points on $a'b'$, see Figure~\ref{fig:sep_region}. 
Let $R$ denote the region of the hyperbolic plane consisting of all points between $ab$ and $a'b'$ whose distance from $tt'$ is at most $\rho$. The resulting shape $R$ is geodesically convex; its boundary consists of two segments ($a_tb_t$ and $a'_tb'_t$), and two hypercycle arcs, denoted by $\arc{a_t b'_t}$ and $\arc{b_t a'_t}$. In general, for two points $u,v$ on one of these hypercycle arcs, let $\arc{uv}$ denote the arc between them, and let $|\arc{uv}|$ be the length of this arc.

Note that any tour segment that connects points on two different sides of $\ell$ also intersects $R$. A tour segment that intersects $R$ can have $0,1$ or $2$ endpoints in $R$. A segment with exactly $1$ endpoint in $R$ is called \emph{entering}. As $R$ is geodesically convex, segments with both endpoints in $R$ are entirely contained in $R$.  All other tour segments crossing $\ell$ must intersect at least one of $\arc{a_t b'_t}$ and $\arc{b_t a'_t}$. We say that a segment \emph{crosses} $R$ if it intersects both $\arc{a_t b'_t}$ and $\arc{b_t a'_t}$. (It is possible that a segment whose endpoints lie outside $R$ on the same side of $\ell$ intersect one of these arcs twice. These segments are not relevant for our algorithm.)

The rest of this section focuses on the following main lemma.

\begin{lemma}\label{lem:R_is_good}
The region $R$ has the following properties:
\begin{itemize}
\item[(i)] $|R\cap P| < n_{in} \eqdef 1+\frac{2(\ln n +1)}{\alpha-2\rho}$
\item[(ii)] There are less than $s_{cr}\eqdef 2+\frac{2(\ln n+1)\cosh \rho}{\rho}$ tour segments that cross $R$.
\end{itemize}
\end{lemma}

The proof requires that we explore the geometry of $R$ more thoroughly.
 
\begin{lemma}\label{lem:Rprop}
We have $|qt|<\ln n + 1$, and $|\arc{a_t b'_t}|=|\arc{b_t a'_t}|<2(\ln n+1) \cosh \rho$. 
\end{lemma}

\begin{proof}
% From the right-angle triangle $qta_t$, we have that $\tan(\an a_tqt)=\frac{\tanh(|ta_t|)}{\sinh(|qt|)}$, so by~\eqref{eq:qt_len}, we have
% \begin{equation}
% \tan(\an a_tqt)= \tan(\frac{\pi}{2n}) \tanh \rho
% \end{equation}
%
We first prove our bound on $|qt|$. Note that $\sinh(.)$ is monotone increasing and $\sinh(|qt|)=\frac{1}{\tan(\frac{\pi}{2n})}$ by~\eqref{eq:qt_len}, so it suffices to show that $\sinh(\ln n +1)>\frac{1}{\tan(\frac{\pi}{2n})}$. Indeed, 
\[\sinh(\ln n+1)=\frac{en-\frac{1}{en}}{2}>n\quad\text{ and }\quad\frac{1}{\tan(\frac{\pi}{2n})}< \frac{1}{\frac{3}{2n}}< n.\]

The arc length of the equidistant hypercycle of base $b$ and distance $\rho$ is $b\cosh \rho$ according to~\cite{Smogor}, therefore $|\arc{a_t b'_t}|=|tt'|\cosh(\rho)=2|qt|\cosh(\rho)<2(\ln n+1) \cosh \rho$. 
\end{proof}

\paragraph*{Ruling out dense crossings}
Our next ingredient for the proof is to show that if two segments cross $R$ very close to each other, then they cannot both be in an optimal tour. Figure~\ref{fig:reroute} illustrates the following lemma. 

\begin{lemma}\label{lem:reroute}
Let $p_1p_2\dots p_ip_{i+1}\dots p_{n-1}p_n$ be an optimal tour on $P$ where both $p_1p_2$ and $p_ip_{i+1}$ cross $R$, and where $p_1,p_i$ and $a_t$ lie on the same side of $\ell$. Let $p'_1=p_1p_2\cap \arc{a_t b'_t}$, and define $p'_2,p'_i$ and $p'_{i+1}$ analogously. Then $|\arc{p'_1p'_i}|+|\arc{p'_2p'_{i+1}}|\geq 4\rho$.
\end{lemma}

\begin{figure}[t]
\centering
\includegraphics[width=\textwidth]{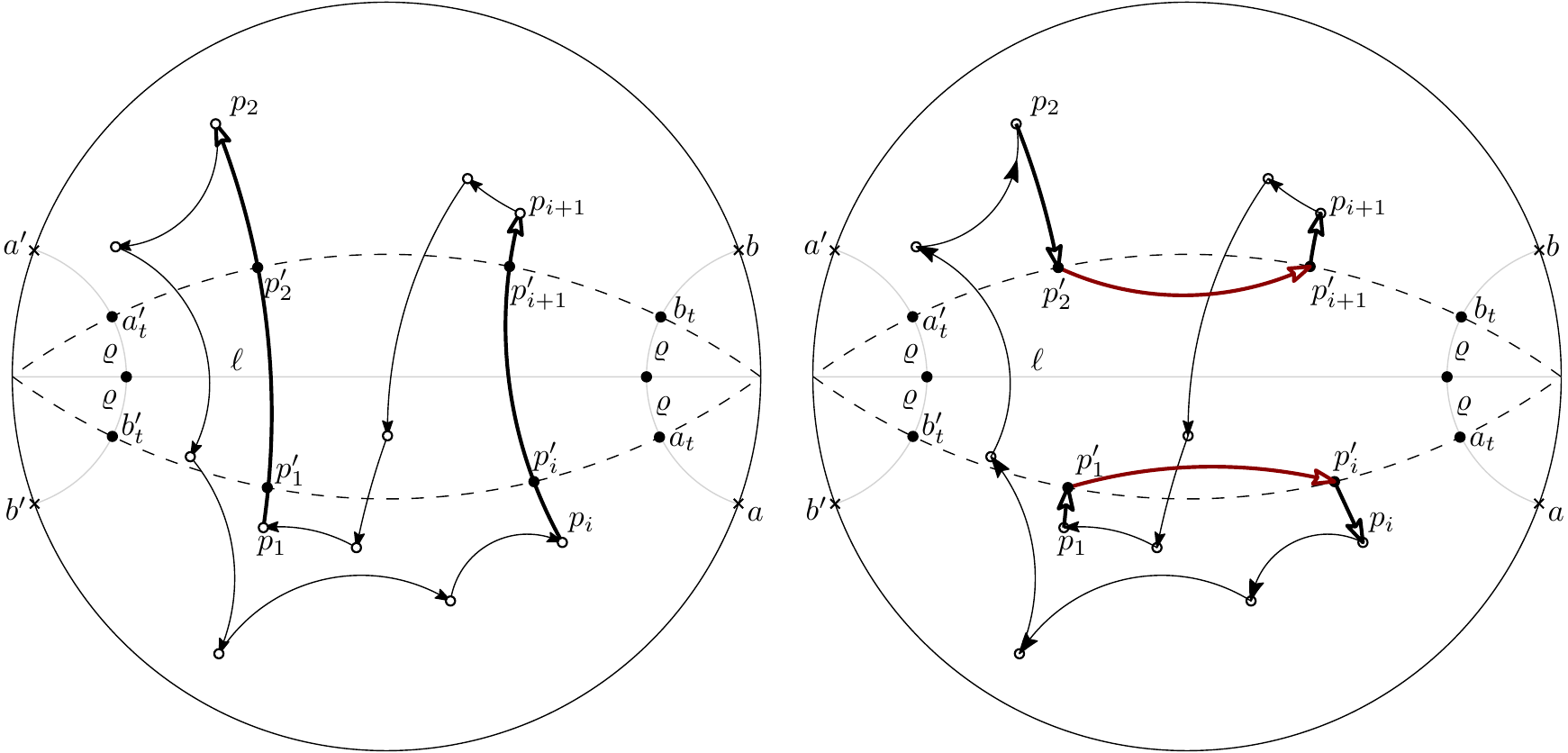}
\caption{Rerouting two crossing edges  ($p_1p_2$ and $p_ip_{i+1}$) into a different tour.}\label{fig:reroute}
\end{figure} 

\begin{proof}
We can create a new tour by removing the segments $p'_1p'_2$ and $p'_ip'_{i+1}$, and replacing them with $p'_1p'_i$ and $p'_2p'_{i+1}$, see Figure~\ref{fig:reroute}. The resulting tour is
\[p_1p'_1p'_ip_ip_{i-1}p_{i-2}\dots p_2p'_2p'_{i+1}p_{i+1}p_{i+2}\dots p_n.\]
Note that this tour contains all the input points.\footnote{This is generally not an optimal tour as it can be further shortened into $p_1p_ip_{i-1}p_{i-2}\dots p_2p_{i+1}p_{i+2}\dots p_n$.} Since the only difference between the tours is that $p'_1p'_2$ and $p'_ip'_{i+1}$ are only present in the optimal tour and $p'_1p'_i$ and $p'_2p'_{i+1}$ are only present in the new tour, by the optimality of $p_1\dots p_n$ we have that
\[0 \geq |p'_1p'_2| + |p'_ip'_{i+1}| - |p'_1p'_i| - |p'_2p'_{i+1}|.\]
Note that $|p'_1p'_2|\geq2\rho$ by the definition of $R$, and analogously $|p'_ip'_{i+1}|\geq2\rho$. Therefore we have
\[
0\geq |p'_1p'_2| + |p'_ip'_{i+1}| - |p'_1p'_i| - |p'_2p'_{i+1}| \geq 4\rho -|\arc{p'_1p'_i}| - |\arc{p'_2p'_{i+1}}|,
\]
which concludes the proof.
\end{proof}

We can now prove Lemma~\ref{lem:R_is_good}.

\begin{proof}[Proof of Lemma~\ref{lem:R_is_good}.]
(i) For a point $p\in P\cap R$, let $p_\ell$ denote the point on $\ell$ for which $pp_\ell$ is perpendicular to $\ell$. Let $p,p'\in P\cap R$ be points such that $p_\ell,p'_\ell$ are consecutive on $\ell$ (i.e., there is no $p''\in P\cap R$ such that $p''_\ell\in p_\ell p'_\ell)$.
By the triangle inequality, $|pp_\ell|+|p_\ell p'_\ell|+|p'_\ell p'|\geq |pp'|$, and $|pp'|\geq \alpha$ since $P$ is $\alpha$-spaced. By the definition of $R$ and $\rho$, we also have that $|pp_\ell|\leq\rho$ and $|p'_\ell p'|\leq \rho$. Consequently, 
\begin{equation}
|p_\ell p'_\ell| \geq \alpha-2\rho.
\end{equation}

We can apply this inequality to all consecutive pairs $p_\ell p'_\ell$. Since all the points $p_\ell$ lie on the segment $tt'$, the total length of the segments $p_\ell p'_\ell$ cannot exceed $|tt'|$. It follows that
\[|P\cap R| \leq 1+\left\lfloor\frac{|tt'|}{\alpha-2\rho}\right\rfloor
< 1+\frac{2(\ln n +1)}{\alpha-2\rho},\]
where the second inequality uses our bound from Lemma~\ref{lem:Rprop}.

(ii) Let $p_1\dots,p_n$ be an optimal tour, and let $p_ip_{i+1}$ be an edge crossing $R$. (Indices are defined modulo $n$.) Note that $p_ip_{i+1}$ can cross $R$ in two directions: either from the side of $a$ to the side of $b$ or the other way around. By Lemma~\ref{lem:reroute}, consecutive crossings $p_ip_{i+1}$ and $p_jp_{j+1}$ in the same direction use at least a total arc length of $4\rho$ on the arcs $\arc{a_t b'_t}$ and $\arc{b_t a'_t}$. Since the total length of these arcs is less than $4(\ln n+1) \cosh \rho$ by Lemma~\ref{lem:Rprop}, the number of crossings in one direction is less than
\[1+\left\lfloor\frac{4(\ln n+1)\cosh \rho}{4\rho}\right\rfloor\leq 1+\frac{(\ln n+1)\cosh \rho}{\rho}.\]
Consequently, the total number of crossings (in both directions) is less than 
\[2+\frac{2(\ln n+1)\cosh \rho}{\rho}.\]
This concludes the proof.
\end{proof}

\section{A divide-and-conquer algorithm}

In order for a divide-and-conquer approach to work for \ETSP, one should be able to solve subproblems with partial tours. We follow the terminology and definitions of De~Berg~\etal~\cite{BergBKK18} here.
Let $M$ be a perfect matching on a set~$B\subseteq P$ of 
so-called \emph{boundary points}. We say that a 
collection~$\cP=\{\pi_1,\ldots,\pi_{|B|/2}\}$ of paths \emph{realizes $M$ on $P$}
if (i) for each pair $(p,q)\in M$ there is a path $\pi_i\in \cP$ with~$p$ and~$q$
as endpoints, 
and (ii) the paths together visit each point $p\in P$ exactly once.
We define the length of a path $\pi_i$ to be the sum of the lengths of its edges, and we define the total length of $\cP$ to be the sum of the lengths of the paths~$\pi_i\in\cP$.
The subproblems that arise in our divide-and-conquer algorithm can be defined
as follows. 

%----------------------------------------------------------------------------
\defproblem{\BTSP}
{A point set $P\subset \Hyp^2$, a set of boundary points $B\subseteq P$, and a perfect matching $M$ on $B$.}
{Find a collection of paths of minimum total length that realizes~$M$ on~$P$.}
%----------------------------------------------------------------------------
Let \emph{PathTSP}$(P,B,M)$ be the optimal tour length for the instance $(P,B,M)$. Note that we can solve \HTSP on a point set $P$ by solving \BTSP $n-1$ times on $P$ with
$B:=\{p,q\}$ and $M := \{(p,q)\}$ for each $q\in P\setminus\{p\}$, and answering
\[\min_{q\in P\setminus \{p\}} \bigg(\text{\emph{PathTSP}}\big(P,\{p,q\},\{(p,q)\}\big)+|pq|\bigg).\]

\subsection{Algorithm}
Our algorithm is a standard divide and conquer algorithm that is very similar to~\cite{SmithW98} and~\cite{BergBKK18}. The algorithm requires knowledge of the initial value of $\alpha$; we can compute this before the first call in $O(n^2)$ time. We give a pseudocode (Algorithm~\ref{alg:hTSP}) and also explain the steps below. In the explanation, we sometimes regard sets of segments with endpoints in $P$ as subgraphs of the complete graph with vertex set $P$.

\begin{algorithm}[t]
\noindent
\caption{\emph{HyperbolicTSP}$(P,B,M,\alpha)$}\label{alg:hTSP}
\textbf{Input:} A set $P\subset \Reals^d$, a subset $B\subseteq P$, a perfect matching $M\subseteq \binom{B}{2}$, and initial spacing $\alpha$   \\
\textbf{Output:} The minimum length of a path cover of $P$ realizing the matching $M$ on $B$  \\[-5mm]
\begin{algorithmic}[1]
\If{$|P| \leq t$} \label{step:threshold}
   \Return \emph{BruteForceTSP}$(P,B,M)$
\EndIf 
\If {$|B|<\max\left(\frac{40\ln |P|}{\alpha}, 8\ln |P|\right)$}\label{step:sep}
	\State Compute a centerpoint $q$ of $P$, the line $\ell$ through $q$ and the region $R$.\label{step:balance_normal}
\Else
	\State Compute a centerpoint $q$ of $B$, the line $\ell$ through $q$ and the region $R$.\label{step:balance_trick}
\EndIf 
\State $\Cr \gets \{pp' \mid p,p'\in P,\, pp' \text{ crosses } R\}$,\, $\End \gets \{pp' \mid p\in R\cap P,\, p'\in P, pp'\text{ intersects }\ell\}$
\State $mincost \gets \infty$
\ForAll{$S_{cr}\subseteq \Cr$, $|S_{cr}|\leq s_{cr}$} \label{step:mainloop}
\ForAll{$S_{end}\subseteq \End$, the maximum degree of $S_{end}$ is at most $2$}\label{step:secondloop}
	\State $P_1,P_2\gets$ uncovered vertices on each side of $\ell$
	\State $B_1,B_2\gets$ boundary vertices of $P_1$  (resp., $P_2$).
	\ForAll{perfect matchings $M_1$ on $B_1$ and $M_2$ on $B_2$}\label{step:matchloop}
		\If{$M_1\cup M_2\cup S_{cr} \cup S_{end}$ realize $M$}\label{step:ifrealize}
			\State $c_1\gets HyperbolicTSP(P_1,B_1,M_1,\alpha)$ \label{step:recurse1}
			\State $c_2\gets HyperbolicTSP(P_2,B_2,M_2,\alpha)$\label{step:recurse2}
			\If{$c_1+c_2+\len(S_{cr} \cup S_{end})<mincost$}
				\State $mincost\gets c_1+c_2+\len(S_{cr} \cup S_{end})$
			\EndIf
		\EndIf
	\EndFor
\EndFor
\EndFor
\State \Return{$mincost$}
\end{algorithmic}
\end{algorithm}

As a first step, we run a brute-force algorithm (comparing all path covers of $P$) if the input point set $P$ has size at most the threshold $t$, where $t$ will be a large constant. On line~\ref{step:sep}, we check the size of the boundary. If it is less than $\max\left(\frac{40\ln |P|}{\alpha}, 8\ln |P|\right)$, then we compute the centerpoint of $P$, the line $\ell$ with the empty cone according to Lemma~\ref{lem:sep}, and the region $R$. Otherwise (similarly to~\cite{BergBKK18}), we need to shrink the boundary, so we use a line $\ell$ through the centerpoint of $B$ instead. Next, we define the segment set $\Cr$ as the set of segments $pp'$ that cross $R$, and $\End$ as the set of segments intersecting $\ell$ that have at least one endpoint in $R$. We initialize the returned value $mincost$ to infinity.

On line~\ref{step:mainloop}, we iterate over all segment sets $S_{cr}\subseteq \Cr$ with $|S_{cr}|\leq s_{cr}$, where $s_{cr}$ is our bound on the number of crossing segments from Lemma~\ref{lem:R_is_good}. The algorithm considers $S_{cr}$ to be the set of segments crossing $R$. Next, we iterate over all the sets $S_{end}\subseteq \End$ where each point of $P$ has at most two incident segments from $S_{end}$. The algorithm considers $S_{end}$ to be the set of segments crossing $\ell$ with at least one endpoint in $R$.

Each point in $B$ needs to have one adjacent segment in the optimum tour $\cP$, and each point in $P\setminus B$ needs two such points. We say that a point $p\in B$ (resp., $p\in P\setminus B$) is \emph{uncovered} if its degree in $S_{cr} \cup S_{end}$ is less than $1$ (resp., $2$). We denote by $P_1$ and $P_2$ the set of uncovered points on each side of $\ell$. A point $p\in P$ is a boundary point if $p\in B$ and $p$ has degree $0$ in $S_{cr} \cup S_{end}$, or $p\in P\setminus B$ and it has degree $1$ in $S_{cr} \cup S_{end}$. We let $B_1$ denote the boundary points in $P_1$. Similarly, $B_2$ is the set of boundary points in $P_2$.

Line~\ref{step:matchloop} proceeds by iterating over all perfect matchings $M_1$ on $B_1$ and $M_2$ on $B_2$. If the graph on $B_1\cup B_2\cup B$ formed by $M_1\cup M_2  \cup S_{cr} \cup S_{end}$ is a set of paths such that contracting all edges that are not incident to $B$ results in $M$, then we say that $M_1\cup M_2 \cup S_{cr} \cup S_{end}$ realizes $M$. If this is the case for a particular choice 
$M_1,M_2$, then on lines~\ref{step:recurse1} and~\ref{step:recurse2} we recurse on both $P_1$ and $P_2$. The resulting path covers together with $S_{cr} \cup S_{end}$ form a path cover realizing $M$: if their length is shorter than $mincost$, then we update $mincost$. After the loops have ended, we return $mincost$.

We can also compute the optimum tour itself with a small modification of the algorithm.

\paragraph*{Correctness}
The same algorithmic strategy has been used several times in the literature~\cite{SmithW98,BergBKK18}, so we only give a brief justification. First, notice that the algorithm only returns costs of feasible solutions, therefore the returned value is at least as large as the optimum. It remains to show that the returned cost is less or equal to the optimum. Given an optimal path cover $\cP$, the set $S_{cr}$ of segments in $\cP$ crossing $R$ has size at most $s_{cr}$ by Lemma~\ref{lem:R_is_good}. The set of segments $S_{end}$ with one endpoint in $R$ has degree at most two at each point of $R\cap P$.  Consequently, both sets will be considered in Line~\ref{step:mainloop} and Line~\ref{step:secondloop}. The segments of $\cP$ not in $S_{cr}\cup S_{end}$ form a path cover of $P_1$ and $P_2$ with boundary set $B_1$ and $B_2$. These path covers realize some perfect matchings $M_1$ and $M_2$ on $B_1$ and $B_2$ respectively. The matchings $M_1$ and $M_2$ will be considered in the loop at line~\ref{step:matchloop}, and since $M_1\cup M_2\cup S_{cr}\cup S_{end}$ realizes $M$, the recursive steps on Lines~\ref{step:recurse1} and~\ref{step:recurse2} will be executed. The optimal path covers of $(P_1,B_1,M_1)$ and $(P_2,B_2,M_2)$ together with $S_{cr}\cup S_{end}$ give a path cover of optimal cost.

\subsection{Analyzing the running time}

Any given line of Algorithm~\ref{alg:hTSP} other than the recursive calls and loops can be executed in $O(n^2)$ time. Let us consider the loops next. The number of segment sets $S_{cr}$ to be considered in line~\ref{step:mainloop} is at most $\binom{|\Cr|}{s_{cr}}=O(n^{2 s_{cr}})$, since $|\Cr|=O(n^2)$. The number of segment sets $S_{end}$ to be considered is at most $O(n^{2|R\cap P|})\leq O(n^{2n_{in}})$. By Lemma~\ref{lem:R_is_good}, the loops in line~\ref{step:mainloop} and line~\ref{step:secondloop} together have at most
\begin{equation}\label{eq:mainloop}
O\left(n^{2\left(1+\frac{2(\ln n +1)}{\alpha-2\rho} + 2+\frac{2(\ln n+1)\cosh \rho}{\rho}\right)}\right)=O\left(n^{4(\ln n +1)\left(\frac{1}{\alpha-2\rho}+\frac{\cosh \rho}{\rho}\right) + O(1)}\right)
\end{equation}
iterations. Instead of trying to minimize this expression by our choice of $\rho$, we settle for something that is easy to handle. Let
\[\rho\eqdef \min\left(\frac{3}{10}\alpha,\frac{12}{10}\right).\]
 The exponent of~\eqref{eq:mainloop} can be bounded the following way.
If $\alpha<4$, then $\rho=\frac{3}{10}\alpha$, and\\
  $\cosh(\rho)<1.82$, so we get

\begin{equation}\label{eq:boundarybd1}
\begin{split}
 && 2(s_{cr}+n_{in})&= 4(\ln n +1)\left(\frac{1}{\alpha-2\rho}+\frac{\cosh \rho}{\rho}\right) + O(1)\\
 && &<4(\ln n +1)\left(\frac{1}{\frac{4}{10}\alpha}+\frac{1.82}{\frac{3}{10}\alpha}\right) + O(1)\\
 && &<4(\ln n )\left(\frac{8.57}{\alpha}\right) + O(1/\alpha)\\
 && &<35\frac{\ln n}{\alpha},
\end{split}
\end{equation}

 where the last step uses that $n$ is large enough, which we can ensure by setting the threshold~$t$  in Line~\ref{step:threshold} large enough. If $\alpha\geq 4$, then $\rho=1.2$:

\begin{equation}\label{eq:boundarybd2}
\begin{split}
 && 2(s_{cr}+n_{in})&= 4(\ln n +1)\left(\frac{1}{\alpha-2\rho}+\frac{\cosh \rho}{\rho}\right) + O(1) \\
 && &<4(\ln n +1)\left(\frac{1}{\alpha-2.4}+\frac{1.82}{1.2}\right) + O(1) \\
 && &<7\ln n.
\end{split}
\end{equation}

Next, we will analyze the loop at line~\ref{step:matchloop}, but this will require a bound on the size of the boundary set $B$. The following lemma handles the cases $\alpha<4$ and $\alpha\geq 4$ together.

\begin{lemma}\label{lem:bdbd}
The size of the boundary set $B$ is at most $\max(\frac{60\ln |P|}{\alpha}, 12\ln |P|)$ at every recursion level of Algorithm~\ref{alg:hTSP}.
\end{lemma}

\begin{proof}
The statement holds for the initial call as we have $|B|=2$ and $|P|=n$ there.

Notice that if $|B|<\max(\frac{40\ln |P|}{\alpha}, 8\ln |P|)$, then we use the branch on line~\ref{step:balance_normal}. Consequently, the boundary set $B_1$ (and $B_2$) in the new recursive call always has size at most $|B|+(s_{cr}+n_{in})$. So by induction and the bounds~\eqref{eq:boundarybd1} and~\eqref{eq:boundarybd2}, we have that
\begin{align*}
|B_1|&\leq \max\left(\frac{40\ln |P|}{\alpha}, 8\ln |P|\right)+\max\left(\frac{17.5\ln |P|}{\alpha}, 3.5\ln |P|\right)\\
&<\max\left(\frac{57.5\ln |P|}{\alpha}, 11.5\ln |P|\right)\\
&<\max\left(\frac{60\ln |P_1|}{\alpha}, 12\ln |P_1|\right),
\end{align*}
where we use $|P_1|\geq |P|/3 \Rightarrow \ln(|P|)<\ln(P_1)+1.1$; therefore, the last inequality holds if we set the threshold $t$ large enough.

In case of $|B|\geq \max(\frac{40\ln |P|}{\alpha}, 8\ln |P|)$, we use the branch on line~\ref{step:balance_trick}. We have that $|B_1|\leq \frac{2}{3}|B|+(s_{cr}+n_{in})$. By induction, we still have $|B|<\max(\frac{60\ln |P|}{\alpha}, 12\ln |P|)$, so
\begin{align*}
|B_1|&\leq \frac{2}{3}\max\left(\frac{60\ln |P|}{\alpha}, 12\ln |P|\right)+\max\left(\frac{17.5\ln |P|}{\alpha}, 3.5\ln |P|\right)\\
&<\max\left(\frac{60\ln |P_1|}{\alpha}, 12\ln |P_1|\right).
\qedhere
\end{align*}
\end{proof}

The number of perfect matchings on a boundary set $B_1$ is at most $|B_1|^{O(|B_1|)}$. 
Let $b\eqdef\max(\frac{60\ln |P|}{\alpha}, 12\ln |P|)$ be the bound acquired above. The number of iterations of the loop at line~\ref{step:matchloop} is at most $b^{O(b)}$. If $\alpha\geq 4$, then this is $(\ln|P|)^{O(\ln |P|)}=|P|^{O(\ln\ln|P|)}<|P|^{\epsilon\ln|P|}$ for any $\epsilon>0$, as long as $|P|$ is large enough. If $\alpha<4$, then we get

\[b^{O(b)}=\left(\frac{\ln n}{\alpha}\right)^{O(\frac{\ln n}{\alpha})}=n^{O(\frac{1}{\alpha}(\ln\ln n + \ln(1/\alpha) )}.\]
As long as $1/\alpha=n^{o(1)}$, this term is insignificant compared to the iterations of the other loop. Otherwise, we have $\frac{1}{\alpha}\leq \sqrt{n}$, and therefore 
\[b^{O(b)}=n^{O(\frac{1}{\alpha}(\ln\ln n + \ln(1/\alpha) )}=n^{O(\frac{\log n}{\alpha})}.\]

\begin{remark}
If one wants to optimize the leading coefficient in the exponent of the eventual running time, then it is possible to modify the algorithm to use only $c^{|B_1|}$ matchings for $M_1$ as all other matchings lead to crossings. See for example the technique in~\cite{DeinekoKW06}. As a consequence, the leading coefficient will not be influenced by the second loop at all. However, this effort would be in vain if there exists a significantly better algorithm for $\alpha\leq 1$, say $n^{O(\log n \cdot (1/\alpha))}$ or even $n^{O(1/\alpha)}$, which we cannot rule out yet.
\end{remark}

The following lemma finishes the proof of Theorem~\ref{thm:main}.

\begin{lemma}\label{lem:runtime}
The running time of Algorithm~\ref{alg:hTSP} on our initial call is $n^{O(\log^2 n)\max(1,1/\alpha)}$.
\end{lemma}

\begin{proof}
By the analysis above, the running time for an instance $(P,B,M,\alpha)$ with $|P|=n$ satisfies the following recursion.
\[T(n)\leq n^{O(\max(\frac{\log n}{\alpha},\log n))}T\left(\frac{2}{3}n\right)\]
Therefore, there exists a constant $c$ such that the running time is at most
\begin{align*}
T(n)&\leq n^{\max(1,1/\alpha)\cdot c \log n} \left(\frac{2}{3}n\right)^{\max(1,1/\alpha)\cdot c (\log (\frac{2}{3}n))} \cdot \left(\frac{4}{9}n\right)^{\max(1,1/\alpha)\cdot c (\log (\frac{4}{9}n))} \cdot \dots\\
&=  n^{\max(1,1/\alpha)\cdot c (\log n + \log(\frac{2}{3}n) +\log(\frac{4}{9}n)+\dots )}\\
&=n^{\max(1,1/\alpha)\cdot O(\log^2 n)}.\qedhere
\end{align*}
\end{proof}

\section{Lower bound for point sets with dense point pairs}\label{sec:app_hyptsplower}

In this section we prove the following theorem.

\begin{theorem}\label{thm:lowerHTSP}
If ETH holds, then there is a constant $c>0$ such that there is no $2^{o(\sqrt{n})}$ algorithm for \HTSP on $c/\sqrt{n}$-spaced point sets.
\end{theorem}

We say that a planar digraph $D=(V,A)$ can be drawn in an $n\times n$ grid if we can map its vertices to grid points in the $n\times n$ grid, and we can map each arc $uv\in A$ to vertex disjoint grid paths within the $n\times n$ grid connecting the corresponding grid points.

Let $\cC$ denote the class of directed planar graphs that can be drawn in the $n\times n$ grid,\footnote{Note that the planar graphs in $\cC(n)$ may have $\Omega(n^2)$ vertices.} and at each vertex either the indegree is $1$ and the outdegree is $2$, or the indegree is $2$ and the outdegree is $1$.

The proof is a reduction from directed Hamiltonian cycle in $\cC$.  A middle step in the lower bound for Hamiltonian cycle in~\cite{frameworkpaper} shows that under ETH, there is no $2^{o(n)}$ algorithm for directed Hamiltonian cycle in $\cC(n)$. Our goal is to create a point set $P$ inside $\Hyp^2$ within distance $O(1)$ from the origin, where edges will be represented with ``tentacles'' similarly to the construction of Itai~\etal~\cite{ItaiPS82} for Hamiltonian cycle in grid graphs.

First, given a directed planar graph $G_0\in \cC$ drawn in a Euclidean grid of size $n\times n$, we define a class $\cG$ of undirected planar graphs, each of which has a Hamiltonian cycle if and only if $G_0$ has one. Finally, we show that there is an element $G\in \cG$ and a corresponding $c/n$-spaced point set $P\subset \Hyp^2$ of size $O(n^2)$ such that connecting point pairs of $P$ who are at distance exactly $c/n$ gives a planar drawing of $G$.

\paragraph*{Defining the graph class $\cG$}

Following the terminology of Itai~\etal~\cite{ItaiPS82}, a \emph{strip} is a rectangular grid graph of width $2$, i.e., the Cartesian product of path on at least $3$ vertices with a path on $2$ vertices. See Figure~\ref{fig:stripbend} for an illustration. When drawn in a grid horizontally, a strip consists of some grid squares; the leftmost and rightmost squares are called the \emph{ending squares} of the strip. The edges of the strip that are horizontal in this drawing are called \emph{proper} edges. A \emph{bending} at a square selects a non-ending square with proper edges $uv$ and $u'v'$, contracts the edge $uv$, and subdivides the edge $u'v'$ with a new vertex $w'$. A \emph{tentacle} is a graph obtained from a strip by bending it at some collection of non-adjacent squares. The \emph{ending edges} of a tentacle are the two edges that are induced by vertices who are only incident to the ending squares. Note that we also allow tentacles that are not subgraphs of the Euclidean grid.

\begin{figure}[t]
\centering
\includegraphics{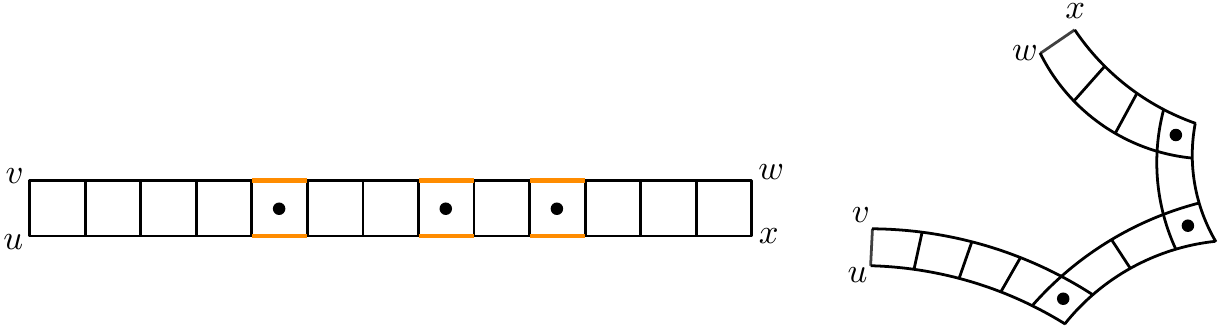}
\caption{Left: A strip with ending edges $uv$ and $wx$, and some proper edges marked for bending. Right: The tentacle created by bending.}\label{fig:stripbend}
\end{figure}

A key observation of~\cite{ItaiPS82} is that a tentacle can be used to represent an edge of $G_0$, as one can ``traverse'' a tentacle in a snake-like fashion to simulate that the Hamiltonian cycle of $G_0$ uses the edge, or it is possible to take a long detour through the tentacle to simulate the fact that the original edge was not contained in the Hamiltonian cycle of $G_0$. %We remark that both the traversal and the detour on the tentacle can be done in two ways.

Let $G_0$ be a directed planar graph of maximum total degree $3$ together with a fixed drawing in an $n\times n$ grid. We replace each arc  of $G_0$ with an arbitrary tentacle drawn in the plane, with ending edges assigned to the source and target of the arc. These tentacles are connected using vertex gadgets, as described below.

\begin{figure}[t]
\centering
\includegraphics{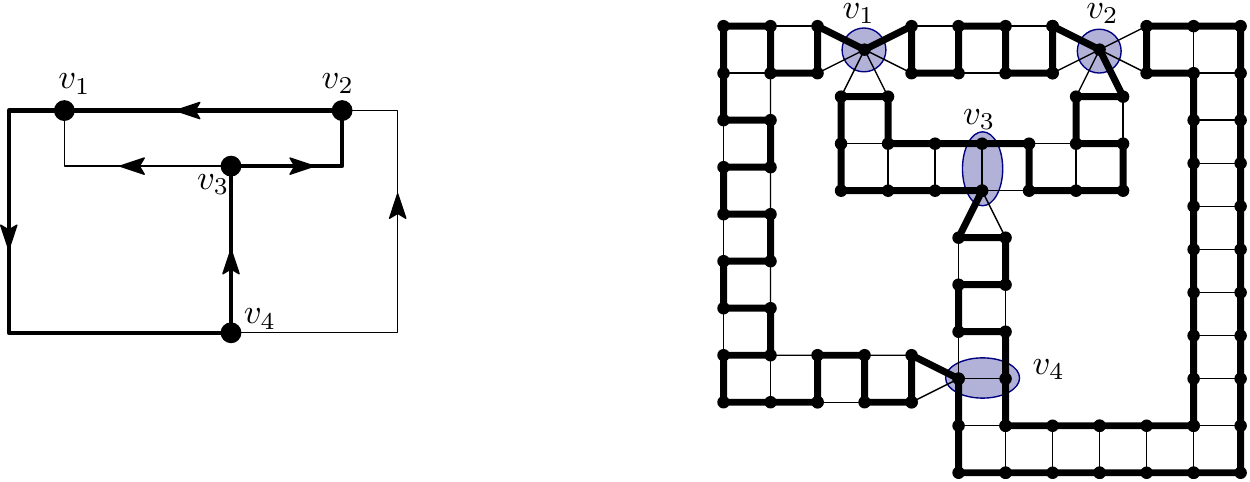}
\caption{A planar graph $G_0\in \cC$ with a directed Hamiltonian cycle, and the corresponding Hamiltonian cycle in a graph $G\in \cG$ with highlighted vertex gadgets.}\label{fig:ham_ex}
\end{figure}

Consider a vertex with indegree $1$. The corresponding vertex gadget has two vertices, $a$ and $b$, and the edge $ab$. The ending edges corresponding to the outgoing arc tentacles are identified with $ab$ , and if the incoming arc tentacle has ending edge $vv'$, then we connect both $v$ and $v'$ to either $a$ or $b$, see Figure~\ref{fig:vertexgadget} for an example. The edge identification and the choice of $a$ or $b$ can be done in a manner that loosely follows the drawing of $G_0$, see Figure~\ref{fig:ham_ex}. We now define the vertex gadget for a vertex of indegree $2$. We create a single vertex, which is connected to the endpoints of the ending edge of each incident arc tentacle, again following the planar drawing of $G_0$.

The planar graphs that can be created in the above manner form the graph class $\cG\eqdef \cG(G_0)$.

\begin{lemma}\label{lem:equiv_G}
For all graphs $G\in \cG$, the graph $G$ has a Hamiltonian cycle if and only if $G_0$ has a directed Hamiltonian cycle.
\end{lemma}

\begin{proof}[Proof sketch.]
If $G_0$ has a directed Hamiltonian cycle, then we can create a Hamiltonian cycle of $G$ the following way. Let $v$ be a vertex of indegree $2$, and let us follow the Hamiltonian cycle in $G$ starting from $v$. If the next vertex on the Hamiltonian cycle is $w$, then we traverse the tentacle connecting the vertex gadgets of $v$ and $w$ in a snake-like manner (see for example the path from $v_1$ to $v_4$ in Figure~\ref{fig:ham_ex}). If $w$ has indegree $2$, then we again continue with a traversal of the outgoing arc tentacle corresponding to the next edge of the Hamiltonian cycle in $G_0$. Otherwise, we have arrived at a vertex $w_a$ in $G$, which is one of the vertices of a gadget that has two outgoing tentacles; let $w_b$ be the other vertex. We make a detour on the outgoing tentacle that corresponds to the arc not used by the Hamiltonian cycle of $G_0$ (not touching the gadget of the other endpoint), and arrive at $w_b$. Then we traverse the other outgoing tentacle starting from $w_b$. We continue the above procedure, until we eventually arrive back to the starting vertex $v$. Note that all vertices in vertex gadgets will be covered, and each tentacle is either traversed or a detour is made on it when we are going through the vertex gadget corresponding to its source.

If $G$ has a Hamiltonian cycle, then all tentacles are either traversed or have a detour on them that starts and ends at their source vertex gadget. One can check that there is indeed no other option for a Hamiltonian cycle to cover a tentacle corresponding to an arc. The arcs corresponding to traversed tentacles give a directed Hamiltonian cycle of $G_0$.
\end{proof}

\paragraph*{Constructing an element of $\cG$ as an $\alpha$-distance graph in $\Hyp^2$}

Our construction will have spacing exactly $\alpha=1/n$. In our figures, we draw an edge between a pair of points if and only if their distance is exactly $\alpha$. If there is no edge, then the distance must be strictly larger than $(1+\eps)\alpha$ for some fixed small $\eps>0$.

\begin{figure}[t]
\centering
\includegraphics{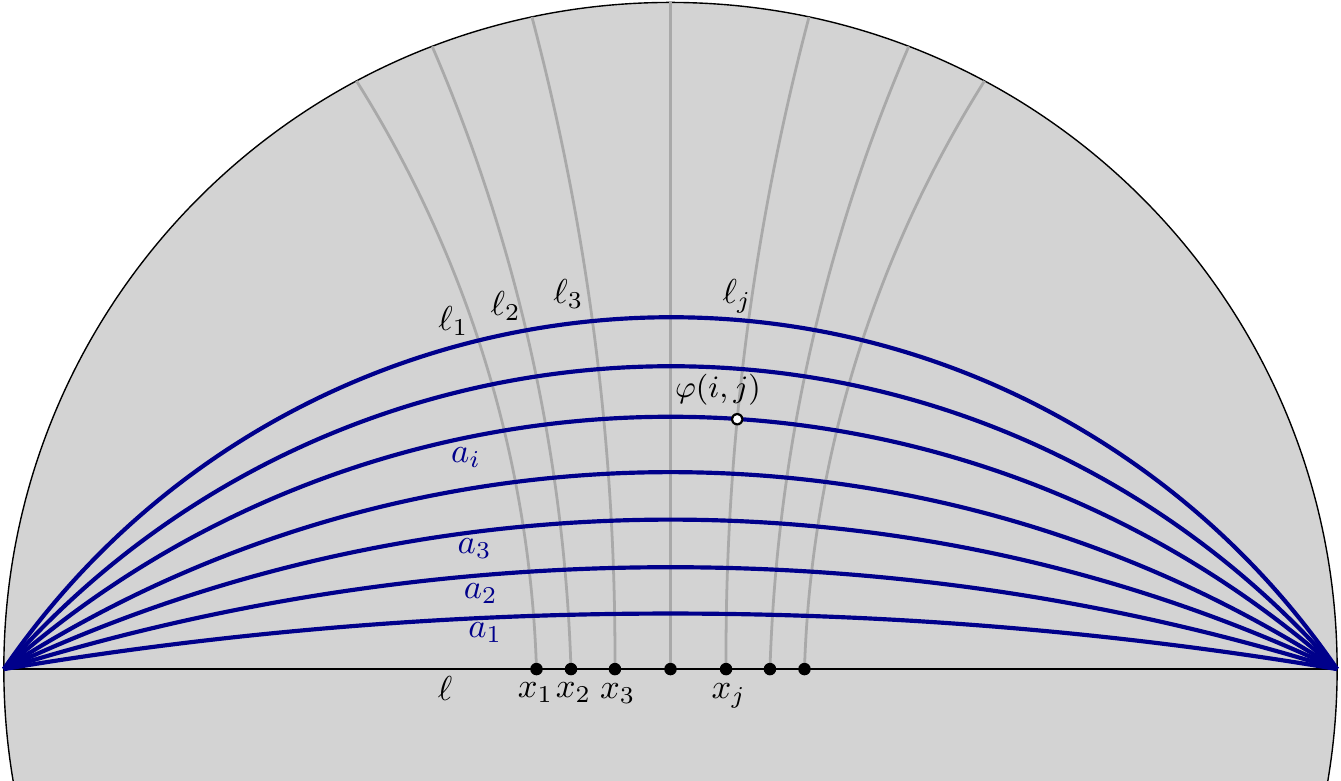}
\caption{A grid-like structure with lines $\ell_j$ and blue hypercycles $a_i$. }\label{fig:poingrid}
\end{figure}

First, we create a grid-like structure of spacing $c\alpha$ for some large constant $c$.
Let $\ell$ be a line, and let $x_1,\dots,x_n$ be points on this line where $\dist(x_j,x_{j+1})=c\alpha$ for $j=1.\dots,n-1$. We draw $n$ lines perpendicular to $\ell$, one through each point $x_i$, denoted by $\ell_1,\dots, \ell_n$, see Figure~\ref{fig:poingrid}. We take $n$ hypercycles on one side of $\ell$, with distances $c\alpha,2c\alpha,\dots,nc\alpha$ from $\ell$, denoted by $a_1,\dots, a_n$. We let $f$ denote the mapping of the $n\times n$ grid to $\Hyp^2$ where $f(i,j)\eqdef a_i\cap \ell_j$.

In the grid-like structure created above, vertically neighboring grid points $(i,j)$ and $(i+1,j)$ are mapped to hyperbolic points at distance exactly $c\alpha$, and they are connected by a line segment of $\ell_j$. For a pair of horizontally neighboring points $(i,j)$ and $(i,j+1)$, they are connected by a hypercycle arc of $a_i$, which has length $c\alpha\cosh(ic\alpha)$. Since $\cosh(x)\geq 1$, this is always at least $c\alpha$. On the other hand, we have
\begin{equation}
c\alpha\cosh(ic\alpha)\leq c\alpha\cosh(nc\alpha)=\Theta(1).\label{eq:gridedge}
\end{equation}

The goal is to follow the drawing of $G_0$ in the grid to realize the tentacles of a graph $G$. The vertex gadget for a vertex $v\in V(G_0)$ that is located at $(i,j)$ will be placed at $\phi(i,j)$, and a grid path corresponding to edge $uv$ will be represented by a tentacle that follows the lines $\ell_j$ and the hypercycles $a_i$ starting at the gadget of $u$ and ending at the gadget of $v$.

Our vertex gadgets are easy to construct, see Figure~\ref{fig:vertexgadget}.
\begin{figure}[t]
\centering
\includegraphics{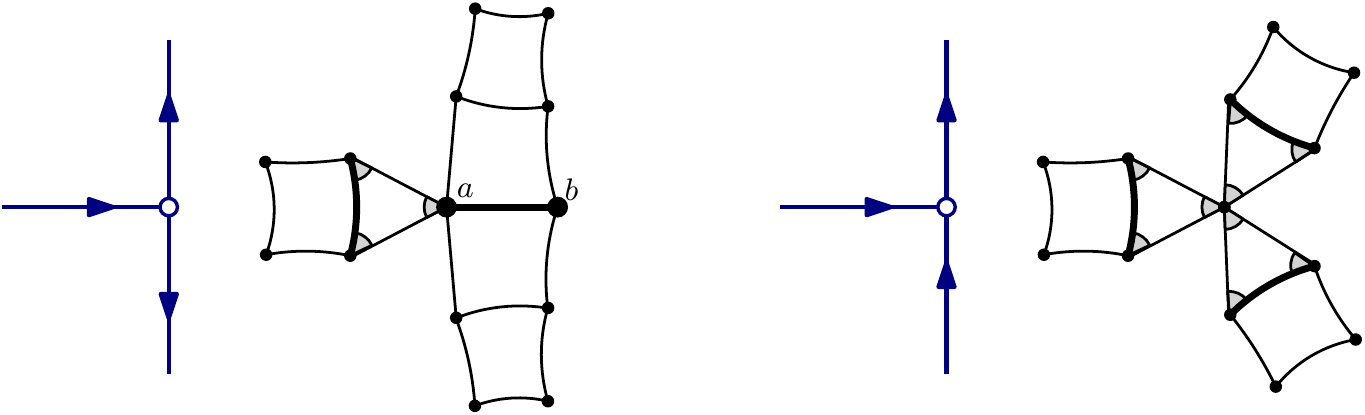}
\caption{Vertex gadgets with connecting tentacles and thick tentacle ending edges. The quadrangles and triangles are regular, and have degree less than $\pi/2$ and less than $\pi/3$ respectively. }\label{fig:vertexgadget}
\end{figure}

The tentacles use quadrangles whose sides have length exactly $\alpha$, and non-adjacent vertices have distance more than $\alpha(1+\eps)$ for some fixed $\eps>0$. A quadrangle with these properties is called \emph{fat}. Note that there is a continuum of fat quadrangles, and we can use this flexibility to adjust our tentacles as needed.

\begin{lemma}\label{lem:adjuststrip}
There exist constants $c_1<c_2$ and $\phi$ such that the following hold.
Let $\ell$ be a line or hypercycle arc, and let $ab$ and $a'b'$ be segments whose midpoints are on $\ell$ at distance $\delta\in (c_1\alpha,c_2\alpha)$, and their line intersects $\ell$ with angle in $[\pi/2-\phi,\pi/2+\phi]$. Then there is a strip of size at most $6$ built of fat quadrangles such that the ending edges are $ab$ and $a'b'$.
\end{lemma}

\begin{figure}[t]
\centering
\includegraphics{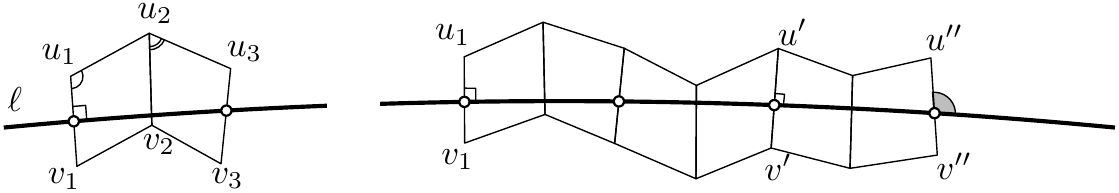}
\caption{Adjusting quadrangles near a line or hypercycle. }\label{fig:barjoint}
\end{figure}

\begin{proof}[Proof sketch]
We regard a strip of fat quadrangles as a system of bars and joints. Let $u_1v_1$ be a segment of length $\alpha$ perpendicular to $\ell$, and consider first a regular quadrangle $u_1v_1v_2u_2$, see Figure~\ref{fig:barjoint}. By increasing or decreasing the angle $v_1u_1u_2$, we can shift the edge $u_2v_2$ ``up'' and ``down''. We attach a regular quadrangle at $u_2v_2$, with vertices $u_2v_2v_3u_3$. If our original shift was small enough, we can decrease the angle $v_2u_2u_3$ until the midpoint of $u_3v_3$ is on $\ell$, and the quadrangle $u_2v_2v_3u_3$ stays fat.

Notice that for any choice of the angle $\an v_1u_1u_2$ within some small interval of the angle of the regular quadrangle, there is a unique corresponding $\an v_2u_2u_3$ for which the midpoint of $u_3v_3$ is on $\ell$. Moreover, the distance of the midpoints of $uv$ and $u_3v_3$ as well as the angle of $u_3v_3$ and $\ell$ are continuous non-constant functions of $\an v_1u_1u_2$. It can be shown that the set of angles between $\ell$ and $u_3v_3$ achievable this way cover some interval $[\pi/2-\phi,\pi/2+\phi]$. With two additional quadrangles, we can get an interval $u'v'$ with its midpoint on $\ell$ that is perpendicular to $\ell$, and whose distance from $uv$ can be varied in some small interval based on our choices for the inner angles of each fat quadrangle. The possible distances attainable between $u_1v_1$ and $u'v'$ cover some interval $[c_1\alpha,c_2\alpha]$, where $c_1< c_2$.

By adding two further quadrangles after $u'v'$, we get a strip of $6$ fat quadrangles where we can also customize the angle of the ending edge $u''v''$. we can construct a sequence of $6$ fat rectangles with ending edges $u_1v_1$ and $u''v''$ that have the desired properties.
\end{proof}

As a corollary, we can create a sequence of fat quadrangles between a pair of ending edges, as long as those ending edges are far enough from each other, while staying in a small neighborhood of a line or hypercycle arc.

\begin{corollary}\label{cor:getstrip}
There exists a constant $c$ such that the following holds.
Let $ab$ and $a'b'$ be segments of length $\alpha$ within distance at most $\alpha$ from $p$ and $q$ respectively. Suppose that $p$ and $q$ are on a line or hypercycle arc $\ell$, where $|pq|=x$ (respectively, $|\arc{pq}|=x$). Then there is a strip consisting of $O(x/\alpha)$ fat quadrangles with ending edges $ab$ and $a'b'$ whose quadrangles are within distance $10\alpha$ of $pq$ (respectively $\arc{pq}$).
\end{corollary}

We are now ready to define our construction.

\begin{lemma}\label{lem:lower_const}
There is an $\alpha$-spaced point set $P$ of size $O(n^2)$ in $\Hyp^2$ such that the graph given by connecting point pairs at distance exactly $\alpha$ has a Hamiltonian cycle if and only if $G_0$ has a Hamiltonian cycle. Moreover, given $G_0$ with its grid drawing, we can  create a point set $P'$ with a word-RAM machine where for each $p\in P$ there is a unique $p'\in P'$ such that $\dist_{\Hyp^2}(p,p')< \frac{\alpha}{n^3}$ in $O(n^2)$ time.
\end{lemma}

\begin{proof}
Consider a grid point $(i,j)$ in the drawing of $G_0$ that corresponds to a vertex $v$ of~$G_0$. Depending on the indegree of $v$, we place a corresponding vertex gadget at $f(i,j)$ (that is, for indegree $2$, the vertex is placed at $f(i,j)$, and for indegree $1$, the midpoint of $ab$ is placed at $f(i,j)$). If $(i,j)$ is a point where an edge of $G_0$ has a right-angle turn, then we place a regular quadrangle $uvwx$ such that the midpoints of $uv$ and $wx$ are on $\ell_j$ and the midpoints of $vw$ and $xu$ are on $a_i$. We call these quadrangles corresponding to right-angle turns of $G_0$ \emph{bend quadrangles}.

Finally, consider a straight segment in the drawing of $G_0$ from $(i,j)$ to $(i,j')$ that is part of an arc's path, such that both $(i,j)$ and $(i,j')$ correspond to either a vertex of $G_0$ or a right-angle turn of the arc. By Corollary~\ref{cor:getstrip}, we can connect the corresponding edge of the vertex gadget or bend quadrangle with a strip that stays in a small neighborhood of $a_i$. We can analogously connect vertices or bends of coordinates $(i,j)$ to $(i',j)$ in a small neighborhood of $\ell_j$. If we choose $c=21$, then the embedded grid has distance at least $21\alpha$ between neighboring $a_i$ and $\ell_j$, and the above constructed strips will remain disjoint from each other (and from non-incident vertex gadgets and bend quadrangles). Putting the strips and bend quadrangles together, we get a tentacle for each edge of $G_0$, therefore the resulting graph $G$ is a member of~$\cG$.

The size of the construction is $O(n^2)$, as each grid edge is represented by $O(1)$ quadrangles. Using $O(\log n)$ bits to represent the coordinates of each point in the Poincar\'e model, we can follow the above construction to get an approximate point set $P'$ as required in $O(n^2)$ time.
\end{proof}

\begin{proof}[Proof of Theorem~\ref{thm:lowerHTSP}]
Let $G_0$ be a directed planar graph in $\cC$ with a given drawing in the $n\times n$ grid. Based on the drawing, we invoke Lemma~\ref{lem:lower_const}, which results in a graph $G$ and a corresponding set $P'$, where each edge of $G$ corresponds to a point pair at distance at least $\alpha(1-2/n^3)$ and at most $\alpha(1+2/n^3)$, and any pair of points in $P'$ that are not connected in $G$ have distance at least $\alpha(1+\eps-2/n^3)$. Consequently, the resulting set is $\alpha'=\alpha(1-2/n^3)=\Theta(1/n)$-spaced. Since we have $|P'|=O(n^2)$, the spacing is $\alpha'=\Theta(\sqrt{|P'|})$ as required.

We claim that there is a TSP tour in $P'$ of length at most $n^2\alpha(1+2/n^3)$ if and only if $G_0$ has a directed Hamiltonian cycle. By Lemma~\ref{lem:equiv_G}, we have that $G_0$ has a directed Hamiltonian cycle if and only if the constructed graph $G$ does. Note that any Hamiltonian cycle of $G$ corresponds to a TSP tour of $P'$ of length at most $n^2\alpha(1+2/n^3)$. A TSP tour of $P'$ that contains $t\geq 1$ segments that do not correspond to edges of $G$ must have length at least $(n^2-t)\alpha(1-2/n^3)+t\alpha(1+\eps-2/n^3)=n^2\alpha(1-2/n^3)+t\eps\alpha$. For $n$ large enough, we have that $\eps>4/n$, therefore
\[n^2\alpha(1-2/n^3)+t\eps\alpha>n^2\alpha(1-2/n^3)+t(4/n)\alpha\geq n^2\alpha(1+2/n^3),\]
thus any tour containing a segment which is not an edge of $G$ is strictly longer than $n^2\alpha(1+2/n^3)$. 

In $O(n^2)$ time, we have created a point set $P'$ with spacing $\Theta(1/\sqrt{|P'|})$ which has a TSP tour of a certain length if and only if there is a directed Hamiltonian cycle in $G_0$. If we could solve TSP on any such $P'$ in $2^{o(\sqrt{|P'|})}$ time, then that algorithm could be composed with the above construction to yield a $2^{o(n)}$ algorithm for directed Hamiltonian cycle in $\cC$, which would contradict ETH.
\end{proof}

\section{Conclusion}

We have devised a separator theorem in $\Hyp^2$ that led to a quasi-polynomial algorithm for \HTSP on constant-spaced point sets. For $\alpha$-spaced point sets with spacing $\alpha\geq \log^2n/\sqrt{n}$ our algorithm runs in $n^{O(\log^2 n)\max(1,1/\alpha)}$ time. When the point set has spacing only $\Theta(\log^2 n/\sqrt{n})$, the algorithm's performance degrades to the point of reaching (roughly) the performance of the Euclidean algorithm. If the point set has even closer point pairs, then the algorithm of Hwang~\etal~\cite{HwangCL93} can be used to obtain a running time of $n^{O(\sqrt{n})}$. We have shown that our algorithm's dependence on density is necessary and for spacing $1/\sqrt{n}$, it cannot be significantly improved under ETH. There are several intriguing questions that are left open. We list some of these questions below.
\begin{itemize}
\item \textbf{Improving the running time, lower bounds.} There is a considerable gap between the running time for \HC in hyperbolic unit disk graphs (which is polynomial) and our \HTSP algorithm, which for constant $\alpha$ runs in $n^{O(\log^2(n))}$ time. Is there an $n^{O(\log n)}$ or a polynomial algorithm for $\alpha\geq 1$? Alternatively, can we prove a (conditional) superpolynomial lower bound? Such a lower bound would have to go beyond the quasi-polynomial lower bound for \IS seen in~\cite{hyperbolic_inters}, as that relies heavily on dense point sets which are not allowed for $\alpha=\Omega(1)$. Another approach would be to use the na\"ive grid embedding of~\cite{hyperbolic_inters} directly, but that does not lead to a superpolynomial lower bound here. 
\item \textbf{Higher dimensions.} The grid-based lower-bound framework of~\cite{frameworkpaper} can be used in $\Hyp^{d+1}$, see~\cite{hyperbolic_inters}. In particular, the ETH-based lower bound of~\cite{BergBKK18} for \ETSP implies that there is no $2^{o(n^{1-1/(d-1)})}$ algorithm for \HTSP in $\Hyp^d$ under ETH. Can we extend our algorithmic techniques to constant-spaced point sets in $\Hyp^d$ and gain algorithms with running time $2^{n^{1-1/(d-1)}\poly(\log n)}$? What happens for denser point sets? As observed in~\cite{BergBKK18}, the techniques of Hwang~\etal~\cite{HwangCL93} do not even seem to extend to $\Reals^d$ for $d\geq 3$. Is a running time of $2^{n^{1-1/d}\poly(\log n)}$ possible for all point sets in $\Hyp^d$?
\item \textbf{A less forgiving parameter.} Our usage of the spacing parameter $\alpha$ may be too restrictive. Is there a better algorithm that can handle more general inputs that can contain a few close point pairs?
\end{itemize}

\bibliography{hypertsp}

\end{document}